\documentclass[11pt,a4 paper]{article}
\newtheorem{theo}{Theorem}[section]

\newtheorem{rem}{Remark}[section]
\newtheorem{defi}{Definition}[section]
\newtheorem{lemma}{Lemma}[section]

\newtheorem{prop}{Proposition}[section]
\newtheorem{ex}{Example}[section]

\usepackage{float}
\usepackage{tocloft}
\usepackage{amsfonts,amsmath,amssymb}
\usepackage{multirow, verbatim}
\usepackage{shadow,enumerate}
\usepackage{makeidx}  
\usepackage{graphicx}
\usepackage{color}
\def\Z{{\mathbb Z}}

\def\00{{\bf 0}}
\def\11{{\bf 1}}
\def\+{\oplus}

\def\whitebox{{\hbox{\hskip 1pt
        \vrule height 6pt depth 1.5pt
        \lower 1.5pt\vbox to 7.5pt{\hrule width
                  3.2pt\vfill\hrule width 3.2pt}%
        \vrule height 6pt depth 1.5pt
        \hskip 1pt } }}
\def\qed{\ifhmode\allowbreak\else\nobreak\fi\hfill\quad\nobreak\whitebox\medbreak}

\begin{document}
\title{Construction methods for generalized bent functions}

\author{
S. Hod\v zi\'c \footnote {University of Primorska, FAMNIT, Koper, Slovenia, e-mail: samir.hodzic@famnit.upr.si}\and
E.~Pasalic\footnote{University of Primorska, FAMNIT \& IAM, Koper, Slovenia, e-mail: enes.pasalic6@gmail.com}
}

\date{}
\maketitle

\begin{abstract}
Generalized bent (gbent) functions is a class of functions $f: \Z_2^n \rightarrow \Z_q$, where $q \geq 2$ is  a positive integer, that generalizes a concept of classical bent functions through their co-domain extension. A lot of research has recently been devoted towards derivation of the necessary and sufficient conditions when $f$ is represented as a collection  of Boolean functions. Nevertheless, apart from the necessary conditions that these component functions are bent   when $n$ is even (respectively semi-bent when $n$ is odd), no general construction method has been proposed yet for $n$ odd case. In this article, based on the use of the well-known Maiorana-McFarland (MM) class of functions, we give an explicit construction method of gbent functions, for any even $q >2$ when $n$ is even and for any $q$ of the form $q=2^r$ (for $r>1$) when $n$ is odd. Thus, a long-term open problem of  providing a general construction method of gbent functions, for odd $n$, has been solved. The method for odd $n$ employs a large class of disjoint spectra semi-bent functions with certain additional properties which may be useful in other cryptographic applications.

\medskip
\noindent
\textbf{Keywords:} Generalized bent functions, Walsh-Hadamard transform, (generalized) Marioana-McFarland class,  Gray maps.

\end{abstract}

\section{Introduction}

A generalization  of Boolean  functions    was introduced  in \cite{Kum85} for considering a much larger class of mappings from $\mathbb{Z}_q^n$ to $\mathbb{Z}_q$ which naturally induced generalized concepts of the well known class of Boolean {\em bent functions} introduced by Rothaus \cite{RO76}. Nevertheless, due to a more natural connection to cyclic codes over rings, functions from $\mathbb{Z}_2^n$ to $\mathbb{Z}_q$, where $q \geq 2$ is a positive integer, have drawn even more attention \cite{KUSchm2007}.
 This  class of mappings $\mathbb{Z}_2^n$ to $\mathbb{Z}_q$ will be called {\em generalized Boolean functions} throughout this article and in particular its subclass possessing similar properties as standard bent functions will be named {\em generalized bent (gbent) functions}. The relations between
generalized bent functions, constant amplitude codes and $\mathbb{Z}_4$-linear codes ($q=4$) were studied in \cite{KUSchm2007}. There are also other generalizations of bent functions such as bent functions over finite Abelian groups for instance \cite{Solod2002}. A nice survey on different generalizations of bent functions can be found in \cite{Tokareva}.

 There are several reasons for studying generalized bent functions. In the first place there is a close connection of these objects to classical bent functions when $n$ is even. Indeed, using a suitable representation of $f: \mathbb{Z}_2^n \rightarrow \mathbb{Z}_q$ as a collection of  its component Boolean functions (whose number depends on 2-adic representation of $q$), it turns out that the necessary  condition for these component functions is that some of their linear combinations are bent if $f$ is supposed to be gbent. The quaternary $q=4$ and octal case $q=8$ were investigated in \cite{Tok} and \cite{Octal}, respectively. Also, in many other recent works \cite{SecCon, OnCross, BentGbent} the authors mainly consider the case $q=2^h$ and  the bent properties of the component functions for a given prescribed form of a gbent function. On the other hand, when $n$ is odd and $q=2^h$, the necessary (but not sufficient) condition that $f$ is gbent is that some linear combinations of the component functions are   semi-bent Boolean functions with the three valued Walsh spectra $\{0,\pm 2^{\frac{n+1}{2}}\}$.

 The main reason, from an applicative point of view, for the interest in these objects is a close relationship between certain objects used in the design of orthogonal frequency-division multiplexing (OFDM)
modulation technique, which in certain cases suffers from relatively high peak-to-mean envelope power ratio
(PMEPR), and gbent functions. To overcome the issues of having large PMEPR, the $q$-ary sequences lying in complementary pairs \cite{Golay} (also called Golay sequences) having a low PMEPR  can be
easily determined from the gbent function associated with this sequence, see \cite{KUGenRM} and the references therein. Another motivation for studying these objects comes from the fact that Gray maps of gbent functions are plateaued functions, see \cite[Propositions 6-7]{SHWM}. The possibility of obtaining plateaued functions from  gbent functions through Gray maps has an independent cryptographic significance.
Thus,  a generic construction of gbent functions also provides a generic methods for designing  plateaued functions by using the results in \cite{SHWM}.

As mentioned above, general construction methods of gbent functions are not known  apart from a few special cases for some particular (small) valued  $q$.  When $q=4$ and $n$ is even, from \cite{Tok} we have that a function $f:\mathbb{Z}^n_2\rightarrow \mathbb{Z}_4$, given in the form $f(x)=a_0(x)+2a_1(x)$, is gbent if and only if $a_1$ and $a_1\oplus a_0$ are  Boolean bent functions.
Several other results related to the case $q=4$ and $n$ even are given in \cite{KUSchm2007}, where some of them involve the trace forms of Galois rings whose employment is also discussed in \cite{XZ}.
For the octal case $q=8$ both necessary and sufficient conditions for the component functions of $f:\Z_2^n \rightarrow \Z_8$,  representing uniquely $f$ as $f(x)=a_0(x)+2a_1(x)+2^2a_2(x)$ where $a_0,a_1,a_2$ are Boolean functions, were given in \cite{BentGbent}. Some recent results on gbent functions related to the case $q=8$  can  be found in \cite{Octal,WM}. Once again, it is necessary (but not sufficient) that certain linear combinations of these Boolean functions are bent when $n$ is even, respectively semi-bent when $n$ is odd. In addition, the Walsh spectra of these functions must satisfy certain conditions related to Hadamard matrices which  makes the design methods rather involved, cf. Theorem~\ref{th1}.

 Several other more general classes of gbent functions were described in \cite{BentGbent}, such as generalized Maiorana-McFarland class (GMMF) \cite[Theorem 8]{BentGbent}, generalized Dillon class (GD)\cite[Theorem 9]{BentGbent}, partial spread class (PS) \cite{TM3} and generalized spread class (GS) \cite[Theorem 10]{BentGbent}. It has been shown that the GD and GMMF classes are both contained in the GS class \cite[Theorem 12]{BentGbent}.
The construction of these gbent functions was also considered in \cite{OnCross} though form the cross-correlation point of view. Apart from the generic construction method of gbent functions  inherent to the GMMF class though only for even $n$, the other classes only provide sufficient gbent conditions which are not easy to satisfy in an efficient manner.
 Gbent functions of the form $g(x)=\frac{q}{2}a(x)+kb(x),$ $k\in \{\frac{q}{4},\frac{3q}{4}\}$, $q=4s$ ($s\in \mathbb{N}$), were analyzed in \cite{SH}, where it has been shown that certain constructions of gbent functions for $q\in\{4,8\}$ \cite{OnCross,Octal,BentGbent} belong to this class of functions (see \cite[Section 5]{SH}). One may notice that many coordinate functions of the function $g$, when $g$ is written in the form (\ref{eq:1}), are equal to each other or possibly are zero functions.
In difference to this approach our construction method  can generate gbent functions for any even $q$ whose pairwise  coordinate functions are different (see Remark \ref{tau_i}), which implies that many gbent functions which are not of the form $\frac{q}{2}a(x)+kb(x)$ can be generated.

 However, the first general characterization of gbent functions, in terms of the choice of component functions for any even $q$  and regardless of the parity of $n$, was given in \cite[Theorem 4.1]{SH2}.
Based on the necessary and sufficient conditions, which are derived in \cite{SHWM}, in this article we present the fist generic  method for construction of gbent functions for any even $q$ when $n$ is even and for $q=2^r$ when $n$ is odd. The method is based on the use of the Maiorana-McFarland (MM) class of functions which contains both semi-bent and bent functions. Nevertheless, the difficulty lies in the fact that the component functions (more precisely certain linear combinations of them) apart from being bent or semi-bent (depending on the parity of $n$) must satisfy additional constraints. More precisely, when $n$ is odd certain linear combinations of the component functions must be disjoint spectra semi-bent functions and apart from that the signs of their Walsh coefficients are supposed to satisfy certain Hadamard recursion, for more details see Section~\ref{sec:descr}. Therefore, the selection of component functions turns out to be a rather nontrivial task. We efficiently solve this problem by using suitable permutations for deriving  disjoint spectra semi-bent functions from the MM class that satisfy the gbent conditions. The question of finding another generic methods for the same purpose is left as an interesting open problem. We emphasize that  the case $n$ even which is also briefly discussed is of minor importance (due to the generic method provided through the GMMF class) and the main contribution is a novel and efficient method of satisfying rather demanding gbent conditions when $n$ is odd.

The rest of this article is organized as follows.   Some basic definitions and notions related to gbent functions are given in Section~\ref{sec:pre}. In Section~\ref{sec:descr} we describe the problem of constructing gbent functions in terms of the sufficient conditions imposed on their component functions.
A method of deriving disjoint spectra semi-bent functions from the MM class, needed in the design of gbent functions for odd $n$, is given in Section~\ref{sec:semibent}, where the case $n$ even is also briefly discussed. In Section~\ref{sec:genconc}, we illustrate construction details for $n$ odd case. Some concluding remarks are found in Section~\ref{sec:conc}.

\section{Preliminaries}\label{sec:pre}

The set of all Boolean functions in $n$ variables, that is the mappings from $\mathbb{Z}_2^n$ to $\mathbb{Z}_2$ is denoted by $\mathcal{B}_n$. Especially, the set of affine functions in $n$ variables we define as $\mathcal{A}_n=\{a\cdot x\oplus b\;|\;a\in\mathbb{Z}_2^n,\; b\in\{0,1\}\}$, where ``$\cdot$'' stands for the standard inner (dot) product of two vectors.
A function $f:\mathbb{Z}^n_2 \rightarrow \mathbb{Z}_2$   is commonly represented using its associated algebraic normal form (ANF) as
\begin{eqnarray}\label{ANF}
f(x_1,\ldots,x_n)=\sum_{u\in \mathbb{Z}^n_2}\lambda_u\displaystyle\prod_{i=1}^n{x_i}^{u_i},
\end{eqnarray}
where the variables $x_i\in \mathbb{Z}_2$, $(i=1, \ldots, n)$, ${\lambda_u \in \mathbb{Z}_2}$, $u=(u_1, \ldots,u_n)\in \mathbb{Z}_2^n$. The \emph{Walsh-Hadamard transform} (WHT) of $f\in\mathcal{B}_n$ at any point $\omega\in\mathbb{Z}^n_2$ is defined by $$W_{f}(\omega)=\sum_{x\in \mathbb{Z}^n_2}(-1)^{f(x)\oplus \omega\cdot x}.$$
\begin{defi}\cite{CCA}\label{def2}
Two Boolean functions $f, g\in \mathcal{B}_{n}$ are said to be a pair of \emph {disjoint spectra functions} if
\begin{equation*}
 W_f(\omega) W_g(\omega)=0, \hspace{0.2cm} \emph {for all} \hspace{0.2cm} \omega\in \mathbb{Z}^n_2.
\end{equation*}
\end{defi}
An $n$-variable function $f$ from $\mathbb{Z}^n_2$ to $\mathbb{Z}_q$, where $q\geq 2$ a positive integer, is called a \emph{generalized Boolean function}  \cite{Tok}. We denote the set of such functions by $\mathcal{GB}^n_q$ and for  $q=2$ the classical Boolean functions in $n$ variables are obtained.
 Let $\zeta=e^{2\pi i/q}$ be a complex $q$-primitive root of unity. The {\em generalized Walsh-Hadamard transform} (GWHT) of $f\in \mathcal{GB}^q_n$ at any point $\omega\in \mathbb{Z}^n_2$ is the complex valued function $$\mathcal{H}_{f}(\omega)=\sum_{x\in \mathbb{Z}^n_2}\zeta^{f(x)}(-1)^{\omega\cdot x}.$$
A function $f\in \mathcal{GB}^q_n$ is called \emph{generalized bent (gbent)} function if $|\mathcal{H}_f(\omega)|=2^{\frac{n}{2}}$, for all $\omega\in \mathbb{Z}^n_2.$ Clearly, when $q=2$, we obtain the  Walsh transform $W_f$ of $f\in\mathcal{B}_n$.

A $(1,-1)$-matrix $H$ of order $p$ is called a \emph{Hadamard} matrix if  $HH^{T}=pI_p,$ where $H^{T}$ is the transpose of $H$, and $I_p$ is the $p\times p$ identity matrix. A special kind of  Hadamard matrix is the \emph{Sylvester-Hadamard} or \emph{Walsh-Hadamard} matrix, denoted by $H_{2^{k}},$ which is constructed using the Kronecker product $H_{2^{k}}=H_2\otimes H_{2^{k-1}},$ where
\begin{eqnarray*}\label{HM}
H_1=(1);\hskip 0.4cm H_2=\left(
                           \begin{array}{cc}
                             1 & 1 \\
                             1 & -1 \\
                           \end{array}
                         \right);\hskip 0.4cm H_{2^k}=\left(
      \begin{array}{cc}
        H_{2^{k-1}} & H_{2^{k-1}} \\
        H_{2^{k-1}} & -H_{2^{k-1}} \\
      \end{array}
    \right).
\end{eqnarray*}

We take that $\mathbb{Z}^n_2$ is ordered as $$\{(0,0,\ldots,0),(1,0,\ldots,0),(0,1,\ldots,0),\ldots,(1,1,\ldots,1)\},$$
and the vector $z_i=(i_0,\ldots,i_{n-1}) \in \mathbb{Z}^n_2$ is uniquely identified by $i \in \{0,1,\ldots,2^n-1\}$.
For a function $g\in \mathcal{B}_n$, the $(1,-1)$-sequence defined by $((-1)^{g(z_0)},(-1)^{g(z_1)},\ldots,(-1)^{g(z_{2^{n}-1})})$ is called the \emph{sequence} of $g$, where $z_i=(i_0,\ldots,i_{n-1}),$ $i=0,1,\ldots,2^{n}-1,$ denotes the vector in $\mathbb{Z}^n_2$ whose integer representation is $i$, that is, $i=\sum_{j=0}^{n-1}i_j 2^j.$

If $2^{p-1}<q\leq 2^p,$ to any generalized function $f:\mathbb{Z}^n_2\rightarrow \mathbb{Z}_q,$ we may associate a unique sequence of Boolean functions $a_i\in \mathcal{B}_n$ ($i=0,1,\ldots,p-1$) such that
\begin{eqnarray}\label{eq:1}
f(x)=a_0(x)+2a_1(x)+2^2a_2(x)+\ldots+2^{p-1}a_{p-1}(x),\; \forall x\in \mathbb{Z}^n_{2}.
\end{eqnarray}
Throughout the article, we will use the well-known fact
\begin{eqnarray}\label{fact}
\sum_{x\in \mathbb{Z}^n_2}(-1)^{w\cdot x}=\left\{\begin{array}{cc}
                                                    2^n, & \text{if}\;\;\; w=\textbf{0}_{n} \\
                                                    0, & \text{otherwise}
                                                  \end{array}
\right.,
\end{eqnarray}
where $\textbf{0}_{n}$ denotes the all-zero vector in $\mathbb{Z}^n_2.$

\section{Problem description}\label{sec:descr}

An intensive study of gbent functions  has recently resulted in their complete characterization  when $q$ is a power of two \cite{SHWM} (some  partial results are also given  in \cite{CT,TM,TM2}). Since the analysis of gbent functions provided in \cite{SHWM} is far more extensive than those given in \cite{CT,TM,TM2}, in this section we will mainly refer to the results given  there.  More precisely, using the approach based on Hadamard matrices,  it has been shown that gbent functions from  $\mathbb{Z}^n_2$ to $\mathbb{Z}_{2^k}$ in algebraic sense correspond to affine spaces of bent or semi-bent functions with certain properties, when $n$ is even or odd, respectively (cf. \cite[Section 4]{SHWM}). The problem of providing generic construction methods of gbent functions is therefore closely related to fulfilling these conditions efficiently.  We recall the characterization of gbent functions given in \cite{SHWM} (which  can also be found   in \cite{CT}).
\begin{theo}\cite{SHWM}\label{th1}
Let $f(x) = a_0(x)+\cdots+2^{p-2}a_{p-2}(x)+2^{p-1}a_{p-1}(x) \in\mathcal{GB}_n^{2^p}$, and let
$h_i(x)=a_{p-1}(x) \+ z_i\cdot (a_0(x),\ldots,a_{p-2}(x))$, $i\in[0,2^{p-1}-1]=\{0,1,\ldots,2^{p-1}-1\}$, where $z_i=(i_0,\ldots,i_{p-2})\in \mathbb{Z}^{p-1}_2$.
\begin{itemize}
\item[(i)]
If $n$ is even, then $f$ is gbent if and only if $h_i$ is bent for all $0\le i\le 2^{p-1}-1$, such that
for all $u\in \mathbb{Z}^n_2$,
\begin{equation}\label{H1}
\mathbf{\mathcal{W}}(u) = (W_{h_0}(u),W_{h_1}(u),\ldots,W_{h_{2^{p-1}-1}}(u)) = \pm 2^{\frac{n}{2}}H_{2^{p-1}}^{(r)}
\end{equation}
for some $r$, $0\le r\le 2^{p-1}-1$, depending on $u$.
\item[(ii)]
If $n$ is odd, then $f$ is gbent if and only if $h_i$ is semi-bent for all $0\le i\le 2^{p-1}-1$, such that
for all $u\in \mathbb{Z}^n_2$,
\begin{equation}
\label{H2}
\mathbf{\mathcal{W}}(u) = (\pm 2^{\frac{n+1}{2}}H^{(r)}_{2^{p-2}},\textbf{0}_{2^{p-2}})\quad\mbox{or}\quad
\mathbf{\mathcal{W}}(u) = (\textbf{0}_{2^{p-2}},\pm 2^{\frac{n+1}{2}} H^{(r)}_{2^{p-2}})
\end{equation}
for some $r$, $0\le r\le 2^{p-2}-1$, depending on $u$ ($\textbf{0}_{2^{p-2}}$ is the all-zero vector of length $2^{p-2}$).
\end{itemize}
\end{theo}
\begin{rem}
In Theorem \ref{th1} the condition $(\ref{H1})$ ($n$ is even) means that any vector $\mathcal{W}(u)=(W_{h_0}(u),\ldots,W_{h_{2^{p-1}-1}}(u))$ must be equal to some row (vector) $H^{(r)}_{2^{p-1}}$ of the Hadamard matrix $H_{2^{p-1}}$ multiplied with $\pm 2^{\frac{n}{2}},$ for all $u\in \mathbb{Z}^n_2.$ For odd $n$, the condition $(\ref{H2})$ implies that  the first (alternatively the second) half  of the vector $\mathcal{W}(u)$ is equal to some row of the Hadamard matrix $H_{2^{p-2}}$ multiplied by $\pm 2^{\frac{n+1}{2}}$, whereas the second (alternatively the first) half equals to all-zero vector $\textbf{0}_{2^{p-2}}$.
\end{rem}
The above result implies that the problem of constructing gbent functions is equivalent to finding an affine space of the coordinate functions $\Lambda=a_{p-1}(x)\+ \langle a_0(x),\ldots,a_{p-2}(x)\rangle$ (corresponding to $h_i(x)$) which are all bent (or semi-bent if $n$ is odd) functions and in addition satisfying the relation (\ref{H1}) (alternatively (\ref{H2}) if $n$ is odd). The analysis given in \cite{SHWM} indicates that these properties are not easy to satisfy and a trivial approach is to select most of the coordinate functions to be constant or affinely related to each other. In the extreme case, one may, for even $n$,  specify $a_0(x)=\ldots =a_{p-2}(x)=0$ so that $\Lambda=a_{p-1}(x)$, thus reducing the dimension of $\Lambda$ to be zero.

The difficulty of constructing  gbent functions, thus satisfying  (\ref{H1}) or (\ref{H2}), is closely related to certain equivalent conditions given recently in \cite{SHWM}.
According to \cite[Corollary 2]{SHWM} the relation (\ref{H1}), for even  $n$, can be  equivalently stated as follows:  for any three distinct integers $i,j,k\in \{0,\ldots,2^{p-1}-1\}$, it must hold that $h_i h_j\+ h_i h_k\+ h_j h_k$ is a bent function \footnote{For shortness of notation we usually drop the variables, thus writing $h_i$ instead of $h_i(x)$}, where $h_i,h_j,h_k\in \Lambda$ and the functions $h_l$ are defined as in Theorem~\ref{th1}. Then, the fact that   $h_i h_j\+ h_i h_k\+ h_j h_k$ is bent if and only if $h^*_i\+ h^*_j\+h^*_k=(h_i\+h_j\+h_k)^*$ \cite[Theorem 4]{Sihem} clearly indicates the hardness of the imposed conditions. Indeed, the dual of a sum of bent functions is in general not equal to the sum of duals of these functions, except in the cases when  these functions are affinely related  to each other (thus $h_i=h_j \+ g$, where $g$ is an affine function)  \cite[Proposition 3]{Decom}.
A trivial method for satisfying these conditions, as indicated in \cite[Example 3]{SH2}, is to select certain functions to be constant which then  significantly limits the number of choices and consequently the cardinality of $\mathcal{GB}^n_q$ is quite small.

The case  $n$ being odd appears to be even harder since apart from finding an affine space $\Lambda$ of semi-bent functions, the condition (\ref{H2}) also implicitly involves the disjoint spectra property. More precisely, for any two integers $i\in[0,2^{p-2}-1]$ and $j\in[2^{p-2},2^{p-1}-1]$ it must hold that $W_{h_i}(u)W_{h_j}(u)=0,$ for any $u\in \mathbb{Z}^n_2,$ that is,  $h_i=a_{p-1}\oplus z_i\cdot (a_0,\ldots,a_{p-2})$ and  $h_j=a_{p-1}\oplus z_j\cdot (a_0,\ldots,a_{p-2})$ are disjoint spectra semi-bent functions.
Moreover, as observed in \cite[Example 3]{SH2}, a trivial selection of coordinate semi-bent functions is not possible in this case since specifying some of these coordinate functions to be  constant would violate the equality $W_{h_i}(u)W_{h_j}(u)=0,$ which needs to be satisfied for any two integers $i\in[0,2^{p-2}-1]$ and $j\in[2^{p-2},2^{p-1}-1]$.

The above discussion demonstrates the hardness of the underlying problem and also motivates the need for some efficient and generic construction methods of gbent functions, which is the main objective of this article.
 Since the $n$ odd case appears to be  more difficult then the $n$ even case, we  focus on the construction of semi-bent functions $h_i=a_{p-1}\oplus z_i\cdot (a_0,\ldots,a_{p-2}),$ $i\in[0,2^{p-1}-1],$  satisfying the condition (\ref{H2}) along with the mentioned disjoint spectra property. Even though our proposed construction method  for odd $n$ can be easily adopted to cover  the $n$  even case, the latter case is just briefly mentioned because the GMMF class provides an efficient and generic construction method.

\section{Construction of gbent functions using MM class}\label{sec:semibent}

In this section, we describe an efficient method (based on a subtle employment of the MM class) for specifying disjoint spectra semi-bent functions satisfying the gbent conditions given by (\ref{H2}).

\subsection{Disjoint spectra semi-bent functions in the MM class}

Since our method utilizes the well-known MM-class of functions, we start with the definition of this class. For $x\in \mathbb{Z}^{s}_2$ and $y\in \mathbb{Z}^v_2$, let $g:\mathbb{Z}^{s+v}_2\rightarrow \mathbb{Z}_2$ be defined  as
$$g(x,y)=\phi (x)\cdot y\oplus d(x),$$
where  $\phi:\mathbb{Z}^s_2 \rightarrow \mathbb{Z}^v_2$ and $d\in \mathcal{B}_v$ is an arbitrary function. Then, the function $g$ belongs to the MM-class which can also be represented as a concatenation of affine functions ($g$ is an affine function for any fixed $x$).
It is well-known that if $\phi:\mathbb{Z}^s_2 \rightarrow \mathbb{Z}^v_2$ is injective then the Walsh spectra of $g$ is three-valued and $W_g(u) \in \{0,\pm 2^{v}\}$, for any $u \in \mathbb{Z}^{v+s}_2$. In particular, when $n=2k+1$ is odd then for $v=k$ and  $s=k+1$ the function $g$ is a semi-bent function.

For our purpose, we are interested in finding a set of semi-bent functions such that certain linear combinations of these have the property of being disjoint spectra semi-bent functions. Therefore, we introduce  a useful classification of these functions in terms of disjoint image sets of the mapping $\phi$.
Let $n=2k+1$ be an odd positive integer and $\pi:\mathbb{Z}^k_2\rightarrow \mathbb{Z}^k_2$ be an arbitrary mapping.  We can define $\phi : \mathbb{Z}^k_2 \rightarrow \mathbb{Z}^{k+1}_2$ so that one coordinate is fixed, where without  loss of generality (and to avoid complicated notation) we assume that the first coordinate is fixed so that $\phi_j:\mathbb{Z}^k_2\rightarrow \mathbb{Z}^{k+1}_2$, for $j=0,1$, is defined as:
\begin{equation} \label{eq:phi}
x \stackrel{\phi_0}\mapsto (0,\pi(x)), \;\;\;x \stackrel{\phi_1}\mapsto (1,\pi(x)),
\end{equation}
where $\pi:\mathbb{Z}^k_2\rightarrow \mathbb{Z}^k_2$.
Then,  if $\pi$ is a permutation  the function
\begin{equation}\label{eq:gsemi}
 g^{(j)}_{\pi}(x,y)=\phi_j(x)\cdot y \oplus d(x),\;\;\; x \in \Z_2^k, \;\;\; y \in \Z_2^{k+1},
\end{equation}
is a semi-bent function (since $\phi_j$ is injective), for $j=0,1$.
Having defined $\phi_j,$ $j\in\{0,1\},$ through the mapping $\pi$ we now introduce two sets that distinguish the semi-bent property with respect to $\pi$,
\begin{eqnarray}\label{Pi}
P^{(j)}_{n}=\{g^{(j)}_{\pi}:\mathbb{Z}^{k}_2\times\mathbb{Z}^{k+1}_2\rightarrow \mathbb{Z}_2\; \mid d(x)=0\;\;\text{and}\;\; \pi \;\;\text{is a permutation on}\;\;\mathbb{Z}^{k}_2\},
\end{eqnarray}
and
\begin{eqnarray}\label{Ri}
R^{(j)}_{n}=\{g^{(j)}_{\pi}:\mathbb{Z}^{k}_2\times\mathbb{Z}^{k+1}_2\rightarrow \mathbb{Z}_2\; \;\mid d(x)=0\;\;\text{and}\;\; \pi\;\;\text{is not a permutation on}\;\;\mathbb{Z}^{k}_2\}.
\end{eqnarray}
In the sets $P^{(j)}_{n}$ and $R^{(j)}_{n}$ the functions  $g^{(j)}_{\pi}$ are defined by (\ref{eq:gsemi}), where (for simplicity of notation used later) we assign $d(x)=0$  so that  $g^{(j)}_{\pi}=\phi_j(x)\cdot y,$ for $j\in\{0,1\}.$
For more clarity, we illustrate this method in the following example.
\begin{ex}\label{ex1}
Let us for $n=2k+1=5$ ($k=2$) construct a semi-bent function in $P_5^{(1)}$. We define the mapping $\phi_1(x)=(1,\pi(x))$ for $x\in \mathbb{Z}^2_2$ as
$$\phi_1(00)=(\textnormal{\bf{1}},0,1), \; \phi_1(10)=(\textnormal{\bf{1}},0,0),\; \phi_1(01)=(\textnormal{\bf{1}},1,0),\; \phi_1(11)=(\textnormal{\bf{1}},1,1),$$
where $\pi$ is obviously a permutation on $\mathbb{Z}^2_2$. Taking  $d(x)=0$  in (\ref{eq:gsemi}), the four subfunctions (obtained by fixing $x \in \mathbb{Z}_2^2$) are then:
$$g^{(1)}_{\pi}(0,0,y)=y_0\oplus y_2; \; \;g^{(1)}_{\pi}(1,0,y)=y_0; \; \; g^{(1)}_{\pi}(0,1,y)=y_0\oplus y_1; \; \;g^{(1)}_{\pi}(1,1,y)=y_0\oplus y_1 \oplus y_2.$$
Thus, the function $g^{(1)}_{\pi}(x,y)=\phi_1(x)\cdot y$ belongs to the set $P^{(1)}_{5}.$
\end{ex}

However, the signs of  Walsh coefficients in linear combinations of the coordinate functions are also  of great importance due to the fact that, for any $u \in \mathbb{Z}^{n}_2$, in relation (\ref{H2}) for either the first half of the vector $\mathcal{W}(u)$ it holds that
\begin{eqnarray}\label{e1}
(W_{h_0}(u),\ldots,W_{h_{2^{p-2}-1}}(u))=\pm 2^{\frac{n+1}{2}}H^{(r)}_{2^{p-2}},\;\;\; r\in[0,2^{p-2}-1],
\end{eqnarray}
or alternatively for the second half we have
\begin{eqnarray}\label{e2}
(W_{h_{2^{p-2}}}(u),\ldots,W_{h_{2^{p-1}-1}}(u))=\pm 2^{\frac{n+1}{2}}H^{(r)}_{2^{p-2}},\;\;\; r\in[0,2^{p-2}-1].
\end{eqnarray}
The following result is proved useful in determining the signs of non-zero Walsh coefficients for semi-bent functions in $P^{(j)}_n.$
\begin{prop}\label{pp0}
Let $g^{(j)}_{\pi}=\phi_j(x)\cdot y,$ be an arbitrary semi-bent function in $P^{(j)}_n,$ where $j\in\{0,1\}$, $n=2k+1$, and $\phi_j$ is given by (\ref{eq:phi}). Then, denoting $\omega_2 \in \Z_2^{k+1}$ by $(t,\omega_2') \in \Z_2 \times \Z_2^k$, for $t\in \{0,1\}$, we have
\begin{eqnarray}\label{DS}
W_{g^{(j)}_{\pi}}(\omega_1,\omega_2)=\left\{\begin{array}{cc}
                                        (-1)^{\omega_1\cdot \pi^{-1}(\omega_2')}\;2^{\frac{n+1}{2}}, & t=j \\
                                        0, & t \neq j
                                      \end{array}
\right.,\;\;\; \forall (\omega_1,\omega_2)\in \mathbb{Z}^{k}_2\times \mathbb{Z}^{k+1}_2.
\end{eqnarray}
\end{prop}
\begin{proof}
 For any $(\omega_1,\omega_2)\in \mathbb{Z}^{k}_2\times\mathbb{Z}^{k+1}_2$, the coefficient $W_{g^{(j)}_{\pi}}(\omega_1,\omega_2)$ can be written as
\begin{eqnarray*}\label{walsh}
W_{g^{(j)}_{\pi}}(\omega_1,\omega_2)&=&\sum_{(x,y)\in \mathbb{Z}^{k}_2\times\mathbb{Z}^{k+1}_2}(-1)^{g^{(j)}_{\pi}(x,y)\oplus (x,y)\cdot (\omega_1,\omega_2)}=\sum_{x\in \mathbb{Z}^{k}_2}(-1)^{x\cdot\omega_1}\sum_{y\in \mathbb{Z}^{k+1}_2}(-1)^{g^{(j)}_{\pi}(x,y)\oplus y\cdot\omega_2} \\
&=& \sum_{x\in \mathbb{Z}^{k}_2}(-1)^{x\cdot\omega_1}\sum_{y\in \mathbb{Z}^{k+1}_2}(-1)^{(j,\pi(x))\cdot y \oplus y\cdot\omega_2}=\sum_{x\in \mathbb{Z}^{k}_2}(-1)^{x\cdot\omega_1}\sum_{y\in \mathbb{Z}^{k+1}_2}(-1)^{((j,\pi(x)) \oplus \omega_2)\cdot y}.
\end{eqnarray*}
The last sum equals zero for any $x \in \Z_2^k$, unless $(j,\pi(x)) \oplus \omega_2=0$ in which case the sum equals $2^{k+1}=2^{\frac{n+1}{2}}$. Using the fact that $\pi$ is a permutation, the  condition $(j,\pi(x)) \oplus \omega_2=(j \oplus t,\pi(x) \oplus \omega_2')={\bf 0}$ is satisfied for $t=j$ and a unique $x$ given by $x=\pi^{-1}(\omega_2')$. \qed
\end{proof}
\begin{rem}\label{permutation}
Notice that taking two functions $g^{(j)}_{\pi}, g^{(j)}_{\sigma} \in R^{(j)}_{n}$ so that $\pi,\sigma$ are not permutations, we may still have the property that $\pi \oplus \sigma$ is a permutation in which case $g^{(j)}_{\pi}\oplus g^{(j)}_{\sigma}$ is a semi-bent function.
\end{rem}
Apart from Proposition \ref{pp0}, one can easily construct disjoint spectra semi-bent functions as follows.
\begin{prop}\label{pp2}
Let $f_{\pi}\in P^{(j)}_{n},$ $j\in\{0,1\}$,  and $g_{\sigma}$ belong either to $P^{(1)}_{n}$ or to $R^{(1)}_{n}$.  If $\pi\oplus \sigma$ is a permutation on $\mathbb{Z}^k_2$, then $f_{\pi}\oplus g_{\sigma}$ is a semi-bent function and the functions $f_{\pi}$ and $f_{\pi}\oplus g_{\sigma}$ are disjoint spectra semi-bent functions.
\end{prop}
\begin{proof}
If $\pi\oplus \sigma$ is a permutation on $\mathbb{Z}^k_2$, then clearly functions $f_{\pi}$ and $f_{\pi}\oplus g_{\sigma}$ are semi-bent functions, since $f_{\pi}\in P^{(j)}_{n}$ and $f_{\pi}\oplus g_{\sigma}$ is given as
$$ f_{\sigma}(x,y)\oplus g_{\pi}(x,y)=((i,\sigma(x)) \oplus (j,\pi(x))) \cdot y= ((i \oplus j,\sigma(x) \oplus \pi(x)) \cdot y,$$ for $i,j \in \{0,1\}$.
Furthermore, if $f_{\pi}\in P^{(j)}_{n},$ $j\in\{0,1\}$, and $g_{\sigma}\in P^{(1)}_{n}$ or $g_{\sigma}\in R^{(1)}_{n}$, then $f_{\pi}\oplus g_{\sigma}\in P^{(1\oplus j)}_n.$ The disjoint spectra property follows trivially  from Proposition \ref{pp0}.
\qed
\end{proof}
The primary condition in Theorem \ref{th1}-(ii) is that the component functions $a_0,\ldots,a_{p-2},a_{p-1}\in \mathcal{B}_n$  are selected so that $h_i=a_{p-1}\oplus z_i\cdot (a_0,\ldots,a_{p-2})$ is a  semi-bent function, for any $i\in[0,2^{p-1}-1].$ Especially, when $i=0$ this implies that $a_{p-1}$ has to be  a  semi-bent function, hence it can be chosen from the set $P^{(j)}_{n}$.
Recall that the vector $\mathcal{W}(u)$ at point $u\in \mathbb{Z}^n_2$ is given as
$$\mathcal{W}(u)=(W_{h_0}(u),\ldots,W_{h_{2^{p-2}-1}}(u),W_{h_{2^{p-2}}}(u),\ldots,W_{h_{2^{p-1}-1}}(u)),$$
and accordingly the WHTs of  $h_i$, for $i\in[0,2^{p-1}-1],$ constitute the first half of $\mathcal{W}(u),$ more precisely $(W_{h_0}(u),\ldots,W_{h_{2^{p-2}-1}}(u))$ which does not involve the function $a_{p-2}$. Nevertheless, this function cannot be arbitrary chosen (for instance cannot be constant) since its presence in $h_j$ when $j\in[2^{p-2},2^{p-1}-1]$ directly affects the disjoint spectra property through $W_{h_i}(u)W_{h_j}(u)=0$.

\subsection{Non-trivial selection of component functions, $n$ odd}\label{sec:nontrivial}

%

We now discuss a suitable selection of the  coordinate functions $a_{p-1},a_0,\ldots,a_{p-2}$ from the sets $P^{(j)}_n$ and/or $R^{(j)}_n$. These sets being closely related to mappings over $\mathbb{Z}^k_2$,  to every coordinate function $a_{p-1},a_0,\ldots,a_{p-2}$ we associate the mappings $\sigma,\tau_0,\ldots,\tau_{k-2}: \mathbb{Z}^k_2 \rightarrow \mathbb{Z}^k_2$ as follows:
\begin{eqnarray}\label{a}
a_{p-1}(x,y)=(j_{p-1},\sigma(x))\cdot y,\; \; a_l(x,y)=(j_l,\tau_l(x))\cdot y,\;\;\;(x,y)\in \mathbb{Z}^k_2\times \mathbb{Z}^{k+1}_2,
\end{eqnarray}
where $j_l\in\{0,1\}$ and $l\in[0,p-2]$. Furthermore, let
\begin{eqnarray}\label{pi}
\pi_i=\sigma\oplus z_i\cdot (\tau_0,\ldots,\tau_{p-2}),
\end{eqnarray}
denote linear combinations of $\sigma, \tau_0,\ldots,\tau_{p-2}$, for $i\in[0,2^{p-1}-1]$, where $\pi_i :\mathbb{Z}^k_2 \rightarrow \mathbb{Z}^k_2$.

Henceforth, instead of using the notation $h_i$, we will use a more precise notation $h^{(j)}_{\pi_i}$ which specifies the function $a_{p-1}\oplus z_i\cdot (a_0,\ldots,a_{p-2})$ with respect to relation (\ref{a}), i.e., the functions $h^{(j)}_{\pi_i}=a_{p-1}\oplus z_i\cdot (a_0,\ldots,a_{p-2})$ are given as
$$h^{(j)}_{\pi_i}(x,y)=(j_{p-1}\+z_i\cdot (j_0,\ldots,j_{p-2}),\sigma(x)\+z_i\cdot (\tau_0(x),\ldots,\tau_{p-2}(x)))\cdot y=(j,\pi_i(x))\cdot y,$$
where $(x,y)\in \mathbb{Z}^k_2\times \mathbb{Z}^{k+1}_2$ and $j=j_{p-1}\+z_i\cdot (j_0,\ldots,j_{p-2})\in\{0,1\}$ ($z_i\in \mathbb{Z}^{p-1}_2$).

In order to fulfill the primary condition of Theorem \ref{th1}-(ii), i.e., to have an affine space of semi-bent functions $\Lambda=a_{p-1}\+z_i\cdot (a_0,\ldots,a_{p-2}),$ we will assume that $h^{(j)}_{\pi_i}$ belongs to $P^{(j)}_n$ for all $i\in[0,2^{p-1}-1]$ ($j\in\{0,1\}$).
\begin{rem}\label{diffj}
For arbitrary (fixed) integers $j_0,\ldots,j_{p-1}\in \{0,1\}$, notice that for two different vectors $z_i$ and $z_{i'}$ from $\mathbb{Z}^{p-1}_2$, we may have that $a_{p-1}\oplus z_i\cdot (a_0,\ldots,a_{p-2})\in P^{(j)}_n$ and $a_{p-1}\oplus z_{i'}\cdot (a_0,\ldots,a_{p-2})\in P^{(j')}_n$ with $j\neq j',$ since vectors $z_i$ and $z_{i'}$ are directly employed in $j=j_{p-1}\+ z_i\cdot (j_0,\ldots,j_{p-2})$ and $j'=j_{p-1}\+ z_{i'}\cdot (j_0,\ldots,j_{p-2}).$
\end{rem}
Recall that in relation (\ref{H2}) for any input vector $u\in \mathbb{Z}^n_2$ we have that half of the vector $\mathcal{W}(u)$ is a non-zero vector, and the remaining half is equal to the zero vector $\textbf{0}_{2^{p-2}}.$ Therefore, to satisfy further the relation (\ref{H2}), Proposition \ref{pp0} implies that the integer $j$ in function $h^{(j)}_{\pi_i}$ must be fixed for all $i\in[0,2^{p-2}-1]$ or for all $i\in[2^{p-2},2^{p-1}-1]$ (unlike the case mentioned in Remark \ref{diffj}), depending on vector $u\in\mathbb{Z}^n_2.$ More precisely, let us assume that $j=j_{p-1}\+z_i\cdot (j_0,\ldots,j_{p-2})\in\{0,1\}$ is fixed (the same) in functions $h^{(j)}_{\pi_i}\in P^{(j)}_n$ for all $i\in[0,2^{p-2}-1]$ (with some $j_0,\ldots,j_{p-1}\in \{0,1\}$). For an arbitrary vector $u=(\omega_1,\omega_2)\in \mathbb{Z}^{k}_2\times \mathbb{Z}^{k+1}_2$, where $\omega_2=(t,\omega_2')\in\mathbb{Z}^{k+1}_2,$ $t\in\{0,1\}$, Proposition \ref{pp0} implies  that the first half of the vector $\mathcal{W}(u)$ (in relation (\ref{H2})) is given as
\begin{eqnarray}\label{piinverz}
(W_{h^{(j)}_{\pi_0}}(u),\ldots,W_{h^{(j)}_{\pi_{2^{p-2}-1}}}(u))=\left\{\begin{array}{cc}
                                                                          \pm 2^{\frac{n+1}{2}}((-1)^{\omega_1\cdot \pi^{-1}_{0}(\omega_2')},\ldots,(-1)^{\omega_1\cdot \pi^{-1}_{2^{p-2}-1}(\omega_2')}), & t=j\\
                                                                          \textbf{0}_{2^{p-2}} & t\neq j
                                                                        \end{array}
\right..
\end{eqnarray}
On the other hand, fixing $j'=j_{p-1}\+z_i\cdot (j_0,\ldots,j_{p-2})\in\{0,1\}$ for all the remaining indices $i\in[2^{p-2},2^{p-1}-1],$ the second half of the vector $\mathcal{W}(u)$ is given as
\begin{eqnarray}\label{piinverz2}
(W_{h^{(j')}_{\pi_{2^{p-2}}}}(u),\ldots,W_{h^{(j')}_{\pi_{2^{p-1}-1}}}(u))=\left\{\begin{array}{cc}
                                                                          \pm 2^{\frac{n+1}{2}}((-1)^{\omega_1\cdot \pi^{-1}_{2^{p-2}}(\omega_2')},\ldots,(-1)^{\omega_1\cdot \pi^{-1}_{2^{p-1}-1}(\omega_2')}), & t=j'\\
                                                                          \textbf{0}_{2^{p-2}} & t\neq j'
                                                                        \end{array}
\right..
\end{eqnarray}
The disjoint spectra property in relation (\ref{H2}) is  described through equality $W_{h^{(j)}_{\pi_i}}(u)W_{h^{(j')}_{\pi_l}}(u)=0$, for any two integers $i\in[0,2^{p-2}-1]$  and $l\in [2^{p-2},2^{p-1}-1].$ Obviously, this property is satisfied in relations (\ref{piinverz}) and (\ref{piinverz2}) if and only if it holds that $j'=j\+ 1,$ due to Proposition \ref{pp2}. However, notice that  $j'$ depends on $j$ and  the function $a_{p-2}$, due to the fact that $a_{p-2}$ is present in all functions $h^{(j')}_{\pi_l}$, for $l\in[2^{p-2},2^{p-1}-1]$. In particular, writing the index $l$ as $l=i+2^{p-2}$ it holds that
$$h^{(j')}_{\pi_l}=h^{(j')}_{\pi_{i+2^{p-2}}}=h^{(j)}_{\pi_i}\+ a_{p-2},\;\;\forall i\in[0,2^{p-2}-1],$$
due to the lexicographic ordering of  $\mathbb{Z}^{p-1}_2.$ Hence, the disjoint spectra property is  fulfilled if and only if $h^{(j)}_{\pi_i}\in P^{(j)}_n$ for all $i\in[0,2^{p-2}-1]$, when  $j$ is fixed, and in addition it is necessary to select $a_{p-2}\in P^{(j\+1)}_n$ or $a_{p-2}\in R^{(j\+1)}_n$ so that $h^{(j')}_{\pi_{i+2^{p-2}}}=h^{(j)}_{\pi_i}\+ a_{p-2}$ belongs to $P^{(j\+1)}_n$ ($j'=j\+1$), for all $i\in[0,2^{p-2}-1].$

Assuming that  the disjoint spectra property is satisfied (through a proper selections of $\sigma,\tau_0,\ldots,\tau_{k-2}$), the condition (\ref{H2}) will be fully satisfied if permutations $\pi_0,\ldots,\pi_{2^{p-1}-1}$ (defined by (\ref{pi})) satisfy the relations (\ref{e1}) and (\ref{e2}). In other words, we need to provide a method of construction of these permutations for which in relations (\ref{piinverz}) and (\ref{piinverz2}) it holds that
\begin{eqnarray}\label{HC}
((-1)^{\omega_1\cdot \pi^{-1}_{0+z\cdot 2^{p-2}}(\omega_2')},\ldots,(-1)^{\omega_1\cdot \pi^{-1}_{2^{p-2}-1+z\cdot 2^{p-2}}(\omega_2')})=\pm H^{(r_z)}_{2^{p-2}},
\end{eqnarray}
for both $z=0,1$ and some $0\leq r_z\leq 2^{p-2}-1.$ Firstly, with the following  result we constrain the choice of permutations $\pi_i$ satisfying the  relations (\ref{piinverz}) and (\ref{piinverz2}).
\begin{lemma}\label{Hrow}
Let $\delta_i:\mathbb{Z}^k_2\rightarrow \mathbb{Z}_2,$ for $i=0,\ldots,2^m-1.$ If for a fixed $x\in \mathbb{Z}^k_2$ the equality
$$((-1)^{\delta_0(x)},\ldots,(-1)^{\delta_{2^m-1}(x)})=\pm H^{(r)}_{2^m},$$
holds for some $r\in\{0,\ldots,2^m-1\},$ then there exist $a,b\in \mathbb{Z}^m_2$ so that
\begin{eqnarray}\label{hadamard}
(\delta_0(x),\ldots,\delta_{2^m-1}(x))=(a\cdot (z_0\oplus b),\ldots,a\cdot (z_{2^m-1}\oplus b)).
\end{eqnarray}
\end{lemma}
\begin{proof}
The proof follows from the fact that any row of  $H_{2^m}$ corresponds to a linear function $l_a\in \mathcal{B}_m$, say $l_a(z)=a\cdot z,$ and the minus sign "$-$" is valid for any $b$ such that $a \cdot b=1$.
\qed
\end{proof}
The  result below gives a general method for constructing  permutations $\pi_i$ defined by (\ref{pi}) for which  (\ref{HC}) holds for both $z=0,1$.
\begin{prop}\label{const}
Let the mappings $\sigma,\tau_0,\ldots,\tau_{p-2}:\mathbb{Z}^k_2\rightarrow \mathbb{Z}^k_2$ used in  (\ref{a}) and (\ref{pi}) be defined as
$$\sigma(x)=xS\oplus d,\;\;\tau_c(x)=v^{(c)},\;\;c\in[0,p-2],\;\;\forall x\in \mathbb{Z}^k_2,$$
where $S\in GL(\mathbb{Z}^k_2)$ is an arbitrary matrix in the group of all invertible $k \times k$ binary matrices  and  $d,v^{(c)}\in\mathbb{Z}^k_2$ are arbitrary (fixed) vectors. Then, the relation (\ref{HC}) holds for both $z=0,1$.
%
\end{prop}
\begin{proof}
Let $d,v^{(c)}\in\mathbb{Z}^k_2$ be arbitrary (fixed) vectors and $S\in GL(\mathbb{Z}^k_2)$ be  any invertible matrix. Let also $u=(\omega_1,\omega_2)\in \mathbb{Z}^k_2\times \mathbb{Z}^{k+1}_2$ be an arbitrary vector, where $\omega_2=(t,\omega'_2)$ ($t\in\{0,1\}$). W.l.o.g. we only consider  the case $z=0$ in (\ref{HC}) (which corresponds to (\ref{e1})), since the same arguments apply  to the case $z=1$ (which corresponds to  (\ref{e2})). Equivalently, $z=0$ means that we are considering the case when $t=j$ (the first equation in  (\ref{piinverz})).

Being  linear permutations on  $\mathbb{Z}^k_2,$  the inverse of $\pi_i(x)=xS\oplus d\+ z_i\cdot (v^{(0)},v^{(1)},\ldots,v^{(p-2)})$ is given as
\begin{eqnarray}\label{piinv}
\pi^{-1}_i(x)=(x\+d\oplus z_i\cdot (v^{(0)},v^{(1)},\ldots,v^{(p-2)}))S^{-1},\;\;\;\forall i\in[0,2^{p-1}-1],\;\;\forall x\in \mathbb{Z}^k_2.
\end{eqnarray}
Hence, using (\ref{piinv}) and denoting by $a=(\omega_1\cdot v^{(0)}S^{-1},\ldots,\omega_1\cdot v^{(p-2)}S^{-1})\in \mathbb{Z}^k_2$ and $b=\omega_1\cdot (\omega'_2\+ d)S^{-1}\in \{0,1\}$, it is not difficult to see that for any $i\in[0,2^{p-1}-1]$ the term $\omega_1\cdot \pi^{-1}_i(\omega'_2)$, which occurs in (\ref{piinverz}) and (\ref{piinverz2}), for any $\omega'_2\in \mathbb{Z}^k_2$ can be written as
$$\omega_1\cdot \pi^{-1}_i(\omega'_2)=a\cdot z_i\+ b,$$
 Consequently, Lemma \ref{Hrow} implies that
$$((-1)^{\omega_1\cdot \pi^{-1}_0(\omega'_2)},\ldots,(-1)^{\omega_1\cdot \pi^{-1}_{2^{p-2}-1}(\omega'_2)})=(-1)^{b}((-1)^{a\cdot z_0},\ldots,(-1)^{a\cdot z_{2^{p-2}-1}})=\pm H^{(r)}_{2^{p-2}},$$
for some $0\leq r\leq 2^{p-2}-1,$ which means that relation (\ref{HC}) holds for $z=0$. Using the same arguments, the relation (\ref{HC}) also holds for $z=1$, which completes the proof.\qed
\end{proof}
\begin{rem}\label{tau_i}
One may notice that in Proposition \ref{const}, if $p-1>2^k$ then some mappings $\tau_i=v^{(i)}\in \mathbb{Z}^k_2$ will be the same (assuming $p$ is fixed in (\ref{eq:1})). However, if $p-1\leq 2^k$ then all mappings $\tau_i$ can be defined to be pairwise different. Moreover, for $p-1\leq k$  the affine space $\Lambda=a_{p-1}\+ \langle a_0,\ldots,a_{p-2}\rangle$ may have the full dimension $p-1$ if the vectors $v^{(0)},\ldots,v^{(p-2)}\in \mathbb{Z}^{k}_2$ constitute a basis of $\mathbb{Z}^k_2$.
\end{rem}
The results/discussions from this subsection allow us to formalize the generic construction method for gbent functions, which is given with the following steps.\\\\
\textbf{Construction 1:} Let $f:\mathbb{Z}^{n}_2\rightarrow \mathbb{Z}_{2^p}$ be defined by (\ref{eq:1}), where $n=2k+1$ ($k\in \mathbb{N}$) and $p\geq 2,$ and let the coordinate functions $a_0,\ldots,a_{p-1}$ be defined by (\ref{a}). The function $f$ is  gbent if its coordinate functions are selected as follows:
\begin{enumerate}[(1)]
\item Select the  corresponding permutations $\sigma,\tau_0,\ldots,\tau_{p-2}$ as defined  in Proposition \ref{const}.
\item With respect to the previous step, set $a_{p-1}\in P^{(j)}_n$ for any $j\in\{0,1\}$, $a_{0},\ldots,a_{p-3}\in R^{(0)}_n$ and $a_{p-2}\in R^{(1)}_n.$
\end{enumerate}
\begin{rem}
Note that the first construction step above ensures that $\Lambda=a_{p-1}\+ \langle a_0,\ldots,a_{p-2}\rangle$ is an affine space of semi-bent functions, for which (\ref{e1}) and (\ref{e2}) are satisfied. The second step ensures the disjoint spectra property in relation (\ref{H2}), thus all functions $a_{p-1}\+ z_i\cdot (a_0,\ldots,a_{p-2})\in P^{(j)}_n$ for all $i\in[0,2^{p-2}-1]$ and $a_{p-1}\+ z_l\cdot (a_0,\ldots,a_{p-2})\in P^{(j\+ 1)}_n$ for all $l\in[2^{p-2},2^{p-1}-1].$
\end{rem}

\subsection{The construction when $n$ is even}

In general, our method of constructing  gbent functions for $n$ odd,  summarized in \textbf{Construction 1}, heavily  relies on Propositions \ref{pp2} and \ref{const}. Nevertheless, assuming that the coordinate functions $a_0,\ldots,a_{p-1}$ (and thus the function $f$ given by (\ref{eq:1})) are defined on $\mathbb{Z}^k_2\times\mathbb{Z}^k_2$ implies that the  $n$  even case can be treated quite similarly.
 Indeed,  considering Proposition \ref{const} as a method of selecting the coordinate functions $a_0,\ldots,a_{p-1}$, then all functions $h_i=a_{p-1}\+ z_i\cdot (a_0,\ldots,a_{p-2})$ (now defined on $\mathbb{Z}^k_2\times\mathbb{Z}^k_2$) will belong to the MM-class of bent Boolean functions, since $a_{p-1}(x,y)=\sigma(x)\cdot y$ is a bent function, and $a_{c}(x,y)=\tau_c(x)\cdot y=v^{(c)}\cdot y$, where $v^{(c)}\in \mathbb{Z}^k_2$ and $c\in[0,p-2]$, are linear functions. The resulting gbent function $f$, given  as
$$f(x,y)=v^{(0)}\cdot y+2v^{(1)}\cdot y+\ldots+2^{p-2}v^{(p-2)}\cdot y+2^{p-1}\sigma(x)\cdot y=g(y)+2^{p-1}\sigma(x)\cdot y,$$
will belong to the GMMF-class of gbent functions. Note that  in \cite[Prposition 1]{SH2} it has been shown that all functions within the GMMF-class satisfy the condition (\ref{H1}).

\section{Illustrating the construction details - an example}  \label{sec:genconc}

In what follows, we illustrate the use of construction steps in \textbf{Construction 1} for providing an example of a gbent function, for odd $n$. Hence, let us consider a generalized function $f:\mathbb{Z}^5_2\rightarrow \mathbb{Z}_{32}$ ($n=5=2k+1$, $q=32$) given as
$$f(x)=a_0(x)+2a_1(x)+4a_2(x)+8a_3(x)+16a_4(x).$$
Recall that the function $f$ is gbent (for $n$ odd) if and only if the set $\Lambda=a_4\+\langle a_0,\ldots,a_3\rangle$ is an affine space of semi-bent functions satisfying (\ref{H2}) (see Theorem \ref{th1}).
Since $k=2$, let $\sigma,\tau_0,\ldots, \tau_3:\mathbb{Z}^2_2\rightarrow \mathbb{Z}^2_2$ correspond to the component functions $a_4,a_0,\ldots,a_3\in \mathcal{B}_5$, respectively. Using Proposition \ref{const}, we define these component functions via $\sigma,\tau_i$  so that  $f$ is a gbent function, as follows:
\begin{eqnarray*}
\sigma(x)&=&x\+ (0,1),\;\;\tau_0(x)=v^{(0)}=(1,0),\;\;\tau_1(x)=v^{(1)}=(0,1),\\
\tau_2(x)&=&v^{(2)}=(0,0),\;\;\tau_3(x)=v^{(3)}=(1,1),
\end{eqnarray*}
for every $x\in \mathbb{Z}^2_2$. Note that the permutation $\sigma(x)=xS \oplus  d$ uses the identity matrix $S$. Thus we complete the first step of \textbf{Construction 1}. Consequently, the coordinate functions are defined as
$$a_4(x,y)=(1,\sigma(x))\cdot y,\;\;a_i(x,y)=(0,\tau_i(x))\cdot y,\;\; i=0,1,2,$$
$$a_3(x,y)=(1,\tau_3(x))\cdot y,\;\;\;(x,y)\in \mathbb{Z}^2_2\times \mathbb{Z}^3_2.$$
Clearly, we have that $a_4\oplus z_i\cdot (a_0,\ldots,a_3)\in P^{(1)}_5$ for $i\in[0,7]$ and $a_4\oplus z_i\cdot (a_0,\ldots,a_3)\in P^{(0)}_5$ for $i\in[8,15],$ $z_i\in \mathbb{Z}^4_2$, thus satisfying the disjoint spectra property (the choice of $a_i$ is in accordance to the second step in \textbf{Construction 1}).
Denoting $W_{h_{\pi_i}}(u)=W_{a_{4}\oplus z_i\cdot (a_0,\ldots,a_3)}(u),$ for $u\in \mathbb{Z}^5_2,$ the vectors $\mathcal{W}(u)=(W_{h_{\pi_0}}(u),\ldots,W_{h_{\pi_{15}}}(u))$ are given in Table \ref{tab1}.
\begin{table}[H]
\scriptsize
\centering
\caption{Vectors $\mathcal{W}(u)$ for all $u\in \mathbb{Z}^5_2.$}
\vskip 0.25cm
\begin{tabular}{|c|c|c|}
  \hline
   $u\in \mathbb{Z}^5_2$ & $\mathcal{W}(u)=(W_{h_{\pi_0}}(u),\ldots,W_{h_{\pi_{15}}}(u))$ & $\mathcal{W}(u)=\{\textbf{0}_{2^3},\pm 8 H^{(r)}_{2^3}\}\;\;\text{or}\;\; W^T=\{\pm 8 H^{(r)}_{2^3},\textbf{0}_{2^3}\}$ \\ \hline \hline
  $u_0$ & $\{0, 0, 0, 0, 0, 0, 0, 0, 8, 8, 8, 8, 8, 8, 8, 8\}$ & $\{\textbf{0}_{2^3},\;8H^{(0)}_{2^3}\}$ \\ \hline
  $u_1$ & $\{0, 0, 0, 0, 0, 0, 0, 0, -8, 8, -8, 8, -8, 8, -8, 8\}$ & $\{\textbf{0}_{2^3},\;-8H^{(1)}_{2^3}\}$ \\ \hline
  $u_2$ & $\{0, 0, 0, 0, 0, 0, 0, 0, 8, 8, -8, -8, 8, 8, -8, -8\}$ & $\{\textbf{0}_{2^3},\;8H^{(2)}_{2^3}\}$ \\ \hline
  $u_3$ & $\{0, 0, 0, 0, 0, 0, 0, 0, -8, 8, 8, -8, -8, 8, 8, -8\}$ & $\{\textbf{0}_{2^3},\;-8H^{(3)}_{2^3}\}$ \\ \hline
  $u_4$ & $\{8, 8, 8, 8, 8, 8, 8, 8, 0, 0, 0, 0, 0, 0, 0, 0\}$ & $\{8H^{(0)}_{2^3},\;\textbf{0}_{2^3}\}$ \\ \hline
  $u_5$ & $\{8, -8, 8, -8, 8, -8, 8, -8, 0, 0, 0, 0, 0, 0, 0, 0\}$ & $\{8H^{(1)}_{2^3},\;\textbf{0}_{2^3}\}$ \\ \hline
  $u_6$ & $\{-8, -8, 8, 8, -8, -8, 8, 8, 0, 0, 0, 0, 0, 0, 0, 0\}$ & $\{-8H^{(2)}_{2^3},\;\textbf{0}_{2^3}\}$ \\ \hline
  $u_7$ & $\{-8, 8, 8, -8, -8, 8, 8, -8, 0, 0, 0, 0, 0, 0, 0, 0\}$ & $\{-8H^{(3)}_{2^3},\;\textbf{0}_{2^3}\}$ \\ \hline
  $u_8$ & $\{0, 0, 0, 0, 0, 0, 0, 0, 8, 8, 8, 8, 8, 8, 8, 8\}$ & $\{\textbf{0}_{2^3},\;8H^{(0)}_{2^3}\}$ \\ \hline
  $u_9$ & $\{0, 0, 0, 0, 0, 0, 0, 0, 8, -8, 8, -8, 8, -8, 8, -8\}$ & $\{\textbf{0}_{2^3},\;8H^{(1)}_{2^3}\}$ \\ \hline
  $u_{10}$ & $\{0, 0, 0, 0, 0, 0, 0, 0, 8, 8, -8, -8, 8, 8, -8, -8\}$ & $\{\textbf{0}_{2^3},\;8H^{(2)}_{2^3}\}$ \\ \hline
  $u_{11}$ & $\{0, 0, 0, 0, 0, 0, 0, 0, 8, -8, -8, 8, 8, -8, -8, 8\}$ & $\{\textbf{0}_{2^3},\;8H^{(3)}_{2^3}\}$ \\ \hline
  $u_{12}$ & $\{8, 8, 8, 8, 8, 8, 8, 8, 0, 0, 0, 0, 0, 0, 0, 0\}$ & $\{8H^{(0)}_{2^3},\;\textbf{0}_{2^3}\}$ \\ \hline
  $u_{13}$ & $\{-8, 8, -8, 8, -8, 8, -8, 8, 0, 0, 0, 0, 0, 0, 0, 0\}$ & $\{-8H^{(1)}_{2^3},\;\textbf{0}_{2^3}\}$ \\ \hline
  $u_{14}$ & $ \{-8, -8, 8, 8, -8, -8, 8, 8, 0, 0, 0, 0, 0, 0, 0, 0\}$ & $\{-8H^{(2)}_{2^3},\;\textbf{0}_{2^3}\}$ \\ \hline
  $u_{15}$ & $\{8, -8, -8, 8, 8, -8, -8, 8, 0, 0, 0, 0, 0, 0, 0, 0\}$ & $\{8H^{(3)}_{2^3},\;\textbf{0}_{2^3}\}$ \\ \hline
  $u_{16}$ & $ \{0, 0, 0, 0, 0, 0, 0, 0, 8, 8, 8, 8, 8, 8, 8, 8\}$ &  $\{\textbf{0}_{2^3},\;8H^{(0)}_{2^3}\}$ \\ \hline
  $u_{17}$ & $ \{0, 0, 0, 0, 0, 0, 0, 0, -8, 8, -8, 8, -8, 8, -8, 8\}$ & $\{\textbf{0}_{2^3},\;-8H^{(1)}_{2^3}\}$ \\ \hline
  $u_{18}$ & $\{0, 0, 0, 0, 0, 0, 0, 0, -8, -8, 8, 8, -8, -8, 8, 8\}$ & $\{\textbf{0}_{2^3},\;-8H^{(2)}_{2^3}\}$ \\ \hline
  $u_{19}$ & $\{0, 0, 0, 0, 0, 0, 0, 0, 8, -8, -8, 8, 8, -8, -8, 8\}$ & $\{\textbf{0}_{2^3},\;8H^{(3)}_{2^3}\}$ \\ \hline
  $u_{20}$ & $\{8, 8, 8, 8, 8, 8, 8, 8, 0, 0, 0, 0, 0, 0, 0, 0\}$ & $\{8H^{(0)}_{2^3},\;\textbf{0}_{2^3}\}$ \\ \hline
  $u_{21}$ & $\{8, -8, 8, -8, 8, -8, 8, -8, 0, 0, 0, 0, 0, 0, 0, 0\}$ & $\{8H^{(1)}_{2^3},\;\textbf{0}_{2^3}\}$  \\ \hline
  $u_{22}$ & $\{8, 8, -8, -8, 8, 8, -8, -8, 0, 0, 0, 0, 0, 0, 0, 0\}$ & $\{8H^{(2)}_{2^3},\;\textbf{0}_{2^3}\}$ \\ \hline
  $u_{23}$ & $\{8, -8, -8, 8, 8, -8, -8, 8, 0, 0, 0, 0, 0, 0, 0, 0\}$ & $\{8H^{(3)}_{2^3},\;\textbf{0}_{2^3}\}$ \\ \hline
  $u_{24}$ & $ \{0, 0, 0, 0, 0, 0, 0, 0, 8, 8, 8, 8, 8, 8, 8, 8\}$ & $\{\textbf{0}_{2^3},\;8H^{(0)}_{2^3}\}$ \\ \hline
  $u_{25}$ & $\{0, 0, 0, 0, 0, 0, 0, 0, 8, -8, 8, -8, 8, -8, 8, -8\}$ & $\{\textbf{0}_{2^3},\;8H^{(1)}_{2^3}\}$ \\ \hline
  $u_{26}$ & $\{0, 0, 0, 0, 0, 0, 0, 0, -8, -8, 8, 8, -8, -8, 8, 8\}$ & $\{\textbf{0}_{2^3},\;-8H^{(2)}_{2^3}\}$ \\ \hline
  $u_{27}$ & $\{0, 0, 0, 0, 0, 0, 0, 0, -8, 8, 8, -8, -8, 8, 8, -8\}$ & $\{\textbf{0}_{2^3},\;-8H^{(3)}_{2^3}\}$  \\ \hline
  $u_{28}$ & $\{8, 8, 8, 8, 8, 8, 8, 8, 0, 0, 0, 0, 0, 0, 0, 0\}$ & $\{8H^{(0)}_{2^3},\;\textbf{0}_{2^3}\}$ \\ \hline
  $u_{29}$ & $\{-8, 8, -8, 8, -8, 8, -8, 8, 0, 0, 0, 0, 0, 0, 0, 0\}$ & $\{-8H^{(1)}_{2^3},\;\textbf{0}_{2^3}\}$  \\ \hline
  $u_{30}$ & $\{8, 8, -8, -8, 8, 8, -8, -8, 0, 0, 0, 0, 0, 0, 0, 0\}$ & $\{8H^{(2)}_{2^3},\;\textbf{0}_{2^3}\}$ \\ \hline
  $u_{31}$ & $\{-8, 8, 8, -8, -8, 8, 8, -8, 0, 0, 0, 0, 0, 0, 0, 0\}$ & $\{-8H^{(3)}_{2^3},\;\textbf{0}_{2^3}\}$ \\
  \hline
\end{tabular}
\label{tab1}
\end{table}
Consequently, the output values of the gbent function $f$ are given by
$$\{0, 0, 0, 0, 24, 24, 24, 24, 9, 25, 9, 25, 17, 1, 17, 1, 26, 26, 10, 10, 2, 2, 18, 18, 19, 3, 3, 19, 11, 27, 27, 11\}.$$

\section{Conclusions}\label{sec:conc}
In this article we have proposed a generic method for constructing gbent functions $f:\Z_2^n \rightarrow \Z_q$. The method presented here covers the case $n$ even completely since a gbent function can be specified for any even $q>2,$ and any odd $n$ for $q$ being a power of $2$. The problem of finding other methods for constructing gbent functions, different to those presented in this article, is left as an interesting research challenge.\\\\
\noindent
{\bf Acknowledgement.} Samir Hod\v zi\' c  is supported in part by the Slovenian Research Agency (research program P3-0384 and Young Researchers Grant). Enes Pasalic is partly supported  by the Slovenian Research Agency
(research program P3-0384 and research project J1-6720).

\end{document}